\documentclass[11pt]{article}

\usepackage{ifthen}
\newboolean{twocol}
\setboolean{twocol}{false}

\usepackage[utf8]{inputenc}
\usepackage{xspace,amsfonts,amsmath,fullpage,wrapfig}
\usepackage{todonotes} 
\usepackage{hyperref}
\newtheorem{proposition}{Proposition}
\newtheorem{theorem}{Theorem}
\newtheorem{claim}{Claim}
\newenvironment{proof}{\noindent\textbf{Proof\ }}{\hspace*{\fill}$\Box$\medskip}

\begin{document}

\title{Randomized Two-Valued Bounded Delay Online Buffer Management}
\author{Christoph D\"urr\thanks{Sorbonne Universit\'{e}, CNRS, Laboratoire d'informatique de Paris 6, LIP6, Paris, France. Corresponding author: christoph.durr@lip6.fr. https://orcid.org/0000-0001-8103-5333}
\and
	Shahin Kamali\thanks{Department of Computer Science, University of Manitoba, Winnipeg, Canada. shahin.kamali@umanitoba.ca. https://orcid.org/0000-0003-1404-2212}}

\maketitle
\begin{abstract}
In the bounded delay buffer management problem unit size packets arrive online
to be sent over a network link.  The objective is to maximize the total weight
of packets sent before their deadline. In this paper we are interested in
the two-valued variant of the problem, where every packet has either low ($1$)
or high priority weight ($\alpha>1$). We show that the optimal randomized competitive ratio
against an oblivious adversary is $1+(\alpha-1)/(\alpha^2+\alpha)$.
\end{abstract}


\section{Introduction}

Online Buffer Management is a scheduling problem which arises in network routers. Packets arrive online to be sent on a specific link, the goal is to maximize the total weight of sent packets under various constraints. Different models have been studied in the past, the FIFO model, where the router stores pending packets in a limited capacity buffer, and the bounded delay model where packets can stay only limited time in the buffer.

We are particularly interested in the latter model, which can be formalized as the following scheduling problem. Time is partitioned into \emph{time slots}. At each time slot, a (potentially empty) set of  jobs of unit length arrive online on a single machine. Job $j$ arrives at a release time $r_j\in \mathbb N$, has a deadline $d_j\in\mathbb N$ and a weight $w_j\in \mathbb R^+$. At every time slot $t\in\mathbb N$, there is a set of \emph{pending jobs}, which is updated by jobs released at time $t$ and by jobs expiring at time $t$. The algorithm can choose to schedule one of the pending jobs, at every time slot
The goal is to maximize the total weight of the scheduled jobs.  The usual restriction on the online algorithm is that these decisions need to be made without knowledge of future arriving jobs. Its performance is measured by the competitive ratio, which compares the objective value reached by the algorithm
with the objective value of the optimal schedule, maximizing the ratio over all possible instances. The deterministic competitive ratio of the problem is defined as the best competitive ratio among all deterministic online algorithms. As randomization usually helps in these online settings, we  consider the randomized competitive ratio against an oblivious adversary. The
term \emph{adversary} comes from the game theoretical setting of an online problem, played between an online algorithm which needs to make decisions at every time slot and a malicious adversary which generates the instance over time. The adversary is ``oblivious" in the sense that it is not aware of the random choices made by the randomized online algorithm. Other adversarial models exist but are not considered in this paper.

\section{Related work}

This problem has been introduced in 2001, together with a lower bound of
$\varphi \approx 1.618$ on the deterministic competitive ratio
{\cite{hajek_competitiveness_2001} and {\cite{kesselman_buffer_2004} and an upper bound of 2, achieved by
the \textsc{Greedy} algorithm which schedules always the heaviest pending job (ties are broken by executing the most urgent job). This gap between lower and upper bound generated a series of studies of variants of this problem, mostly with restrictions on the deadlines of the job, such as agreeable deadlines (jobs can be ordered by release times and deadlines at the same time) or $s$-uniform or $s$-bounded instances ($d_j-r_j$ is $s$ or at most $s$) for some parameter $s$$>1$. Eventually in SODA 2019, the gap was closed, with a sophisticated algorithm and analysis \cite{vesely_phi-competitive_2019}.   The randomized competitive ratio in the oblivious adversary model of this problem is still open, and the optimal competitive ratio is
known to be between $1.25$ \cite{chin_online_2003} and $1.582$ \cite{bartal_online_2004,chin_online_2006}. 

The paper \cite{chin_improved_2004} considered a variant of the problem, where every job has not unit processing time, but a processing time $p\in\mathbb N$, all jobs are tight in the sense $d_j=r_j+p$, and the algorithm can preempt job executions. In this setting, the authors give a barely random algorithm with competitive ratio 2, and showed that this is optimal for barely random algorithms.  By scaling, this is equivalent to the problem with unit processing time jobs by fractional release times and deadlines.  In our problem these times are restricted to integers, and in this context preemption is of no help to the algorithm.  We refer to \cite{goldwasser_survey_2010} for a survey on packet scheduling, as well as to the more recent introduction of \cite{vesely_phi-competitive_2019} and references therein.

In this paper we consider a particular variant of the bounded delay buffer management
problem, where every job can have either low or high weight, that is
$w_j\in\{1,\alpha\}$ for some parameter $\alpha >1$.  The 2-valued variant has
been studied in \cite{englert_considering_2012} who showed the deterministic competitive 
ratio is $\min\{1 + 1 / \alpha, 1+(\alpha - 1) / (\alpha + 1)\}$.
It is motivated by the fact that some applications, as for example online games, could send packets with two kind of informations: Major updates in the game configuration (such as the arrival or departure of participants), and minor updates (such as not critical moves).

\subsection{Contributions} 

The lower bound of the deterministic competitive ratio for the 2-valued variant can be easily transformed  into a randomized lower bound.  In addition this transformation suggests a randomized online algorithm, which we show to have optimal competitive ratio  $1 +  (\alpha - 1)/(\alpha^2 + \alpha)$.

\section{Preliminaries}
\label{sec:prelim}

Formally an instance of the bounded delay buffer management problem consists of a finite sequence of jobs. Each job $j$ has a unit processing time, a release time $r_j\in \mathbb N$, a deadline $d_j\in \mathbb N$ and a priority weight $w_j\in \mathbb R_+$.  The time line consists of time slots $t\in\mathbb N$. During each time slot, at most one job can be executed. A job is pending at time $t$ if $r_j \leq t < d_j$ and if it has not already been scheduled.   An online algorithm $A$ can schedule at each time slot $t\in \mathbb N$ at most one pending job, and has to make this decision without knowledge of future released jobs.  The objective value of a schedule is the total priority weight of all scheduled jobs.   This value is compared with the objective value of an optimal schedule, and the resulting value is called the competitive ratio.  The optimal schedule can be computed offline by solving a maximum weight bipartite matching problem, matching jobs to time slots.  The competitive ratio of algorithm $A$ is the supremum of its competitive ratio over all instances. The competitive ratio of the problem is the infimum of the competitive ratio over all online algorithms. Usually the online computation paradigm is seen as a game played between the algorithm (making decisions which job to schedule) and the adversary (generating jobs).

If $A$ is a randomized algorithm, then the performance of $A$ is the expected objective value of the produced schedule.  Different adversarial models have been studied in the literature. In this paper we focus on the oblivious adversary, where the input is generated without knowledge of the random choices made by the algorithm.
\ifthenelse{\boolean{twocol}}{\begin{figure}}{\begin{wrapfigure}{R}{5cm}}
  \begin{tikzpicture}[scale=0.8]
\draw[|-|] (0,0) node[left]{job 1} -- node[above]{$1$} (1,0);
\draw[|-|] (0,-1) node[left]{job 2} -- node[above]{$\alpha$} (1,-1)  ;
\draw[-|](1,-1) -- (2,-1);
\draw (0,-2.7) node[left]{\textsc{Greedy}} ;
\draw (0,-3) rectangle node{$\alpha$} (1,-2.4) ;
\draw (0,-3.7) node[left]{\textsc{Opt} on instance 1} ;
\draw (0,-4) rectangle node{$\alpha$} (1,-3.4) ;
\draw (1,-4) rectangle node{$1$} (2,-3.4) ;

\draw (0,-5.7) node[left]{\textsc{PEDF}} ;
\draw (0,-6) rectangle node{$1$} (1,-5.4) ;
\draw (1,-6) rectangle node{$\alpha$} (2,-5.4) ;

\draw (0,-6.7) node[left]{job 3} ;
\draw[|-|] (1,-6.7) -- node[above]{$\alpha$} (2,-6.7);

\draw (0,-7.5) node[left]{\textsc{Opt} on instance 2} ;
\draw (0,-7.8) rectangle node{$\alpha$} (1,-7.2) ;
\draw (1,-7.8) rectangle node{$\alpha$} (2,-7.2) ;

\end{tikzpicture}
\caption{The deterministic lower bound from \cite{englert_considering_2012}.}
\label{fig:det-lb}
\ifthenelse{\boolean{twocol}}{\end{figure}
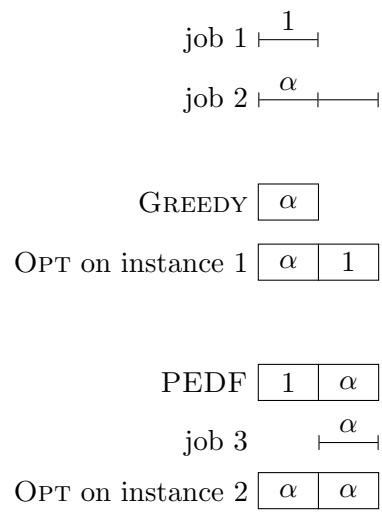}{\end{wrapfigure}{R}{5cm}}

Online algorithms for this problem heavily rely on the notion of a
\emph{provisional schedule}, which was introduced in \cite{englert_considering_2012}. Recall,
that for an algorithm $A$ and a point in time $t$, we define the set of {\em
pending} jobs as the set of jobs released prior or at time $t$, which have
their deadline strictly after $t$, and which have not yet been scheduled by
$A$. The provisional schedule $S$ is a maximum weight subset of pending jobs,
which can be scheduled from time $t$ on.  This set $S$ can be computed
greedily, by starting with $S=\emptyset$, and processing first all pending
heavy jobs, then all pending light jobs, in arbitrary order within each
category. Each considered job $j$ is then added to $S$ if and only if the
number of jobs $i\in S$ with $d_i \leq d_j$ is strictly less than $d_j-t$.
This condition ensures that there is a provisional feasible schedule for $S$.

There are two natural deterministic algorithms that are defined in
\cite{englert_considering_2012}. The first algorithm, which we call
\textsc{Greedy}, schedules at every time step, a job of the provisional
schedule of maximum weight, breaking ties by the deadlines of job (earliest
deadline has the highest priority). The second algorithm, we which we call Provisional Earliest Deadline First (\textsc{PEDF}), schedules at every time step, a job of the
provisional schedule of minimum deadline. Note that it differs from the usual
\textsc{EDF} algorithm (\emph{earliest deadline first}), which is not
restricted to select jobs from the provisional schedule.

The deterministic lower bound consists of an adversary which releases at time
0 two jobs, see Figure~\ref{fig:det-lb}. A job of weight $1$ and deadline $1$
and a job of weight $\alpha$ and deadline 2.  Any deterministic deterministic
algorithm either schedules the urgent job, or the heavy job, or no job at all.
 Clearly the last option is suboptimal, hence we can focus on the first two
options.   If the algorithm schedules the heavy job, then no more jobs are
released, and the competitive ratio is $\frac{1+\alpha}{\alpha}$. If the
algorithm schedules the urgent light job, then at time $1$ the adversary
releases another job of weight $\alpha$ and deadline $2$, leading to the ratio
$\frac{2 \alpha}{1+\alpha}$.  This shows a lower bound of  $\min \{\frac{1 +
\alpha}{\alpha}, \frac{2\alpha}{1+\alpha} \}$ on the deterministic competitive
ratio. It also suggests a simple deterministic algorithm. Let
$\alpha^*=1+\sqrt2$ be the solution to the equation $\frac{1 + \alpha}{\alpha}
= \frac{2\alpha}{1+\alpha}$. If $\alpha \leq \alpha^*$ run
\textsc{PEDF}, else run \textsc{Greedy}.  In
\cite{englert_considering_2012} this strategy was shown to be optimal.

\section{Randomized competitive ratio}
\label{sec:randomized}

The previous lower bound can be extended into a lower bound against an oblivious
randomized adversary. 

\begin{proposition}\label{prop:rand_lower_bound}
No randomized online algorithm can achieve a competitive ratio better than $R=1+\frac{\alpha- 1}{\alpha^2 + \alpha}$.
\end{proposition}

\begin{proof}
Let $\sigma_1,\sigma_2$ be the two instances from the previous lower bound construction. Let $y$ be a distribution on $\sigma_1,\sigma_2$.
We denote ${\mathbb E}_y[\textrm{OPT}(\sigma_j)]$ the expected optimal profit, and by ${\mathbb E}_y[\textrm{A}(\sigma_j)]$ the expected profit of a given randomized algorithm A.
By Yao's principle, see \cite[Theorem 8.3]{borodin_online_2005}, the randomized competitive ratio against an oblivious adversary is lower bounded by 
\[
  \max_y \min_A {\mathbb E}_y[\textrm{OPT}(\sigma_j)] / {\mathbb E}_y[\textrm{A}(\sigma_j)].
\]
We choose $y=(\frac{\alpha - 1}{\alpha}, 1/\alpha)$, which leads to the expected optimal profit
\begin{align*}
  {\mathbb E}_y[\textrm{OPT}(\sigma_j)] = \frac{\alpha-1}\alpha (1+\alpha) + \frac1 \alpha 2 \alpha = 2 + \alpha - 1/\alpha.
\end{align*}
We recall that at time $0$, the instances $\sigma_1$ and $\sigma_2$ are indistinguishable.  Any deterministic algorithm has 3 options at time $0$, to execute the urgent light job, to execute the heavy job or to remain idle.  We ignore the last option, as it is clearly suboptimal.  Any deterministic algorithm, which starts by executing the urgent light job, has an expected profit of at most
\[
    y_1 (1+\alpha) + y_2 (1+\alpha) = 1+\alpha.
\]
In addition any deterministic algorithm, which starts by executing the heavy job, has an expected profit of at most
\[
    y_1 \alpha + y_2 2 \alpha = 
  \frac{\alpha-1}{\alpha} \alpha + \frac1{\alpha} 2 \alpha = 1+ \alpha.
\]
This shows the claimed lower bound.
\end{proof}

Note that the above lower bound is maximized at $\alpha = 1 + \sqrt{2}$ and gives the value $4 - 2 \sqrt{2} \approx 1.172.$

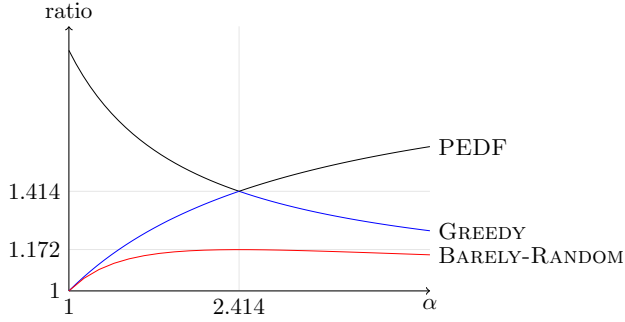
\begin{figure}
\begin{center}
  \scalebox{.8}{
  \begin{tikzpicture}[xscale=2,yscale=4]
\foreach \r in {1,1.172,1.414} {
\draw[color=gray!20] (1,\r) -- (4,\r);
\draw (1,\r) node[left] {\r};
}
\foreach \a in {1,2.414} {
\draw[color=gray!20] (\a,1) -- (\a,2.1);
\draw (\a,1) node[below] {\a};
}
\draw[->] (1,1) -- (4,1) node[below] {$\alpha$}; 
\draw[->] (1,1) -- (1,2.1) node[above] {ratio}; 
\draw [domain=1:2.414,variable=\a] plot (\a, {(\a+1)/\a});
\draw [domain=2.414:4,variable=\a,color=blue] plot (\a, {(\a+1)/\a}) node[right,color=black] {\textsc{Greedy}};
\draw [domain=1:2.414,variable=\a,color=blue] plot (\a, {(2*\a)/(1+\a)});
\draw [domain=2.414:4,variable=\a] plot (\a, {(2*\a)/(1+\a)}) node[right] {\textsc{PEDF}};
\draw [domain=1:4,variable=\a,color=red] plot (\a, {(\a * \a + 2*\a - 1)/(\a * \a + \a)}) node[right,color=black] {\textsc{Barely-Random}};
\end{tikzpicture}}
\end{center}
\caption{Competitive ratios as function of $\alpha$. The deterministic competitive ratio is the minimum of the competitive ratios of the algorithms \textsc{PEDF} and \textsc{Greedy}.}
\label{fig:ratios}
\end{figure}

\subsection{Upper bound}

The lower bound construction suggests a simple barely random algorithm. At
time $0$, the algorithm flips a biased coin and based on its outcome decides between the
two already mentioned deterministic algorithms. In other words,
\textsc{Barely-Random} will run \textsc{Greedy} with probability $p$ and
\textsc{PEDF} with probability $1-p$, for
\[
    p = \frac{\alpha^2 - 1}{\alpha^2 + 2 \alpha -1}.
\]
This surprisingly simple algorithm achieves the optimal competitive ratio.

\begin{theorem}
  Algorithm {\textsc{Barely-Random}} has a competitive ratio of $\frac{\alpha^2 + 2\alpha - 1}{\alpha^2 + \alpha}$ against an oblivious adversary. \label{th:rand_upper}
\end{theorem}

\begin{proof} Fix an instance, and denote by \textsc{Opt} an optimal schedule. 
  We start with showing some structural properties of the three schedules
\textsc{Opt}, \textsc{Greedy} and \textsc{PEDF}. 
Consider a time slot, where \textsc{Greedy} executes a heavy job $j$, while \textsc{PEDF} executes a light job.  By the definition of \textsc{PEDF} we known that \textsc{PEDF} will eventually execute $j$ as well.

In addition we observe that \textsc{Greedy} executes a maximum number of heavy jobs while \textsc{PEDF} executes a maximum number of jobs regardless of their weight.  This follows by the \emph{earliest deadline policy} applied by both algorithms, respectively on heavy pending jobs or on all pending jobs.  Indeed this policy is known to maximize the number of executed jobs for the unweighted unit length job scheduling problem.  
Therefore  \textsc{OPT} does not schedule more jobs than \textsc{PEDF} and does not schedule more heavy jobs than \textsc{Greedy}. Also since \textsc{PEDF} maximizes the number of pending jobs in every time slot, we know that whenever \textsc{PEDF} is idle, then \textsc{Greedy} is idle as well.

As a result, for the analysis of the competitive ratio it suffices to count the number of light and heavy jobs executed by respectively \textsc{Greedy} and \textsc{PEDF}.  
To formalize this intuition, we use the following notation.

\noindent
\scalebox{.9}{
  \begin{tabular}{ll}\\
    set & content \\ \hline
    $D$ & time slots where \textsc{Greedy} is idle but not \textsc{PEDF} \\
    $H^*$ & heavy jobs executed only by \textsc{Greedy} \\
    $H$ & heavy jobs executed both by \textsc{Greedy} and \textsc{PEDF}\\
    $L$ & light jobs executed by \textsc{Greedy} \\
    $L'$ & light jobs executed by \textsc{PEDF} \\
    \hline\\
  \end{tabular}
  }

\begin{figure*}
  \begin{center}
    \resizebox{.95\textwidth}{!}{
  \tikzset{every picture/.style={line width=0.75pt}} 
  \begin{tikzpicture}[x=0.75pt,y=0.75pt,yscale=-1,xscale=1]
    
    \draw  [color=gray, fill=white] (258.72,216.98) -- (274.96,216.98) -- (274.96,194.67) -- (258.72,194.67) -- cycle ;
    \draw  [color=gray, fill=white] (226.26,216.98) -- (242.49,216.98) -- (242.49,194.67) -- (226.26,194.67) -- cycle ;
    \draw  [color=gray, fill=white] (437.27,216.98) -- (453.5,216.98) -- (453.5,194.67) -- (437.27,194.67) -- cycle ;
    \draw  [color=black][fill=gray!35] (193.8,216.98) -- (210.03,216.98) -- (210.03,194.67) -- (193.8,194.67) -- cycle ;
    \draw  [color=black][fill=gray!35] (210.03,216.98) -- (226.26,216.98) -- (226.26,194.67) -- (210.03,194.67) -- cycle ;
    \draw  [fill=gray!80] (177.57,261.6) -- (193.8,261.6) -- (193.8,239.29) -- (177.57,239.29) -- cycle ;
    \draw  [fill=gray!80] (193.8,261.6) -- (210.03,261.6) -- (210.03,239.29) -- (193.8,239.29) -- cycle ;
    \draw  [fill=gray!80] (210.03,261.6) -- (226.26,261.6) -- (226.26,239.29) -- (210.03,239.29) -- cycle ;
    \draw  [fill=gray!80] (226.26,261.6) -- (242.49,261.6) -- (242.49,239.29) -- (226.26,239.29) -- cycle ;
    \draw  [fill=gray!80] (291.19,261.6) -- (307.42,261.6) -- (307.42,239.29) -- (291.19,239.29) -- cycle ;
    \draw  [fill=gray!80] (307.42,261.6) -- (323.65,261.6) -- (323.65,239.29) -- (307.42,239.29) -- cycle ;
    \draw  [fill=gray!80] (339.88,261.6) -- (356.11,261.6) -- (356.11,239.29) -- (339.88,239.29) -- cycle ;
    \draw  [fill=gray!80] (356.11,261.6) -- (372.34,261.6) -- (372.34,239.29) -- (356.11,239.29) -- cycle ;
    \draw  [fill=gray!80] (372.34,261.6) -- (388.57,261.6) -- (388.57,239.29) -- (372.34,239.29) -- cycle ;
    \draw  [fill=gray!80] (388.57,261.6) -- (404.81,261.6) -- (404.81,239.29) -- (388.57,239.29) -- cycle ;
    \draw  [fill=gray!80] (437.27,261.6) -- (453.5,261.6) -- (453.5,239.29) -- (437.27,239.29) -- cycle ;
    \draw  [color=black][fill=gray!35] (388.57,216.98) -- (404.81,216.98) -- (404.81,194.67) -- (388.57,194.67) -- cycle ;
    \draw  [fill=gray!80] (112.64,216.98) -- (128.87,216.98) -- (128.87,194.67) -- (112.64,194.67) -- cycle ;
    \draw  [fill=gray!80] (128.87,216.98) -- (145.1,216.98) -- (145.1,194.67) -- (128.87,194.67) -- cycle ;
    \draw  [fill=gray!80] (145.1,216.98) -- (161.34,216.98) -- (161.34,194.67) -- (145.1,194.67) -- cycle ;
    \draw  [fill=gray!80] (161.34,216.98) -- (177.57,216.98) -- (177.57,194.67) -- (161.34,194.67) -- cycle ;
    \draw  [fill=black!60] (177.57,216.98) -- (193.8,216.98) -- (193.8,194.67) -- (177.57,194.67) -- cycle ;
    \draw  [color=gray, fill=white] (242.49,216.98) -- (258.72,216.98) -- (258.72,194.67) -- (242.49,194.67) -- cycle ;
    \draw  [fill=gray!80] (274.96,216.98) -- (291.19,216.98) -- (291.19,194.67) -- (274.96,194.67) -- cycle ;
    \draw  [fill=black!60] (291.19,216.98) -- (307.42,216.98) -- (307.42,194.67) -- (291.19,194.67) -- cycle ;
    \draw  [fill=gray!80] (307.42,216.98) -- (323.65,216.98) -- (323.65,194.67) -- (307.42,194.67) -- cycle ;
    \draw  [fill=gray!80] (323.65,216.98) -- (339.88,216.98) -- (339.88,194.67) -- (323.65,194.67) -- cycle ;
    \draw  [fill=gray!80] (339.88,216.98) -- (356.11,216.98) -- (356.11,194.67) -- (339.88,194.67) -- cycle ;
    \draw  [fill=gray!80] (356.11,216.98) -- (372.34,216.98) -- (372.34,194.67) -- (356.11,194.67) -- cycle ;
    \draw  [fill=gray!80] (372.34,216.98) -- (388.57,216.98) -- (388.57,194.67) -- (372.34,194.67) -- cycle ;
    \draw  [fill=gray!80] (404.81,216.98) -- (421.04,216.98) -- (421.04,194.67) -- (404.81,194.67) -- cycle ;
    \draw  [fill=gray!80] (421.04,216.98) -- (437.27,216.98) -- (437.27,194.67) -- (421.04,194.67) -- cycle ;
    \draw  [color=gray][fill=gray!30] (161.34,261.6) -- (177.57,261.6) -- (177.57,239.29) -- (161.34,239.29) -- cycle ;
    \draw  [color=gray][fill=gray!30] (145.1,261.6) -- (161.34,261.6) -- (161.34,239.29) -- (145.1,239.29) -- cycle ;
    \draw  [color=gray][fill=gray!30] (404.81,261.6) -- (421.04,261.6) -- (421.04,239.29) -- (404.81,239.29) -- cycle ;
    \draw  [color=gray][fill=gray!30] (128.87,261.6) -- (145.1,261.6) -- (145.1,239.29) -- (128.87,239.29) -- cycle ;
    \draw  [color=gray][fill=gray!30] (112.64,261.6) -- (128.87,261.6) -- (128.87,239.29) -- (112.64,239.29) -- cycle ;
    \draw  [color=gray][fill=gray!30] (242.49,261.6) -- (258.72,261.6) -- (258.72,239.29) -- (242.49,239.29) -- cycle ;
    \draw  [color=gray][fill=gray!30] (258.72,261.6) -- (274.96,261.6) -- (274.96,239.29) -- (258.72,239.29) -- cycle ;
    \draw  [color=gray][fill=gray!30] (274.96,261.6) -- (291.19,261.6) -- (291.19,239.29) -- (274.96,239.29) -- cycle ;
    \draw  [color=gray][fill=gray!30] (323.65,261.6) -- (339.88,261.6) -- (339.88,239.29) -- (323.65,239.29) -- cycle ;
    \draw  [color=gray][fill=gray!30] (421.04,261.6) -- (437.27,261.6) -- (437.27,239.29) -- (421.04,239.29) -- cycle ;
    \draw  [color=gray, fill=white] (497,197.5) -- (517,197.5) -- (517,187.5) -- (497,187.5) -- cycle ;
    \draw  [fill=black!60] (497,204.5) -- (517,204.5) -- (517,214.5) -- (497,214.5) -- cycle ;
    \draw  [fill=gray!80] (497,220.5) -- (517,220.5) -- (517,230.5) -- (497,230.5) -- cycle ;
    \draw  [color=black][fill=gray!35] (497,236.5) -- (517,236.5) -- (517,246.5) -- (497,246.5) -- cycle ;
    \draw  [color=gray][fill=gray!30] (497,252.5) -- (517,252.5) -- (517,262.5) -- (497,262.5) -- cycle ;
    
    \draw (45,199.21) node [anchor=north west][inner sep=0.75pt]  [font=\scriptsize] [align=left] {\textsc{Greedy}};
    \draw (52.22,246.08) node [anchor=north west][inner sep=0.75pt]  [font=\scriptsize] [align=left] {PEDF};
    \draw (520,204) node [anchor=north west][inner sep=0.75pt]  [font=\tiny]  {$H^{*}$};
    \draw (520,188) node [anchor=north west][inner sep=0.75pt]  [font=\tiny]  {$D$};
    \draw (520,221) node [anchor=north west][inner sep=0.75pt]  [font=\tiny]  {$H$};
    \draw (520,236) node [anchor=north west][inner sep=0.75pt]  [font=\tiny]  {$L$};
    \draw (520,252) node [anchor=north west][inner sep=0.75pt]  [font=\tiny]  {$L'$};
    \draw (245.16,264) node [anchor=north west][inner sep=0.75pt]  [font=\scriptsize]  {$x$};
    \draw (263.71,264) node [anchor=north west][inner sep=0.75pt]  [font=\scriptsize]  {$y$};
    \draw (408.67,264) node [anchor=north west][inner sep=0.75pt]  [font=\scriptsize]  {$z$};
    \draw (116.91,267) node [anchor=north west][inner sep=0.75pt]  [font=\scriptsize]  {$a$};
    \draw (133.87,264) node [anchor=north west][inner sep=0.75pt]  [font=\scriptsize]  {$b$};
    \draw (195.87,184) node [anchor=north west][inner sep=0.75pt]  [font=\scriptsize]  {$x$};
    \draw (214.42,184) node [anchor=north west][inner sep=0.75pt]  [font=\scriptsize]  {$y$};
    \draw (391.97,184) node [anchor=north west][inner sep=0.75pt]  [font=\scriptsize]  {$z$};
    
    \draw    (120.76,250.45) -- (120.76,207.83) ;
    \draw [shift={(120.76,205.83)}, rotate = 450] [color={rgb, 255:red, 0; green, 0; blue, 0 }  ][line width=0.75]    (6.56,-1.97) .. controls (4.17,-0.84) and (1.99,-0.18) .. (0,0) .. controls (1.99,0.18) and (4.17,0.84) .. (6.56,1.97)   ;
    \draw    (136.99,250.45) -- (136.99,207.83) ;
    \draw [shift={(136.99,205.83)}, rotate = 450] [color={rgb, 255:red, 0; green, 0; blue, 0 }  ][line width=0.75]    (6.56,-1.97) .. controls (4.17,-0.84) and (1.99,-0.18) .. (0,0) .. controls (1.99,0.18) and (4.17,0.84) .. (6.56,1.97)   ;
    \draw    (153.22,250.45) -- (153.22,207.83) ;
    \draw [shift={(153.22,205.83)}, rotate = 450] [color={rgb, 255:red, 0; green, 0; blue, 0 }  ][line width=0.75]    (6.56,-1.97) .. controls (4.17,-0.84) and (1.99,-0.18) .. (0,0) .. controls (1.99,0.18) and (4.17,0.84) .. (6.56,1.97)   ;
    \draw    (169.45,250.45) -- (169.45,207.83) ;
    \draw [shift={(169.45,205.83)}, rotate = 450] [color={rgb, 255:red, 0; green, 0; blue, 0 }  ][line width=0.75]    (6.56,-1.97) .. controls (4.17,-0.84) and (1.99,-0.18) .. (0,0) .. controls (1.99,0.18) and (4.17,0.84) .. (6.56,1.97)   ;
    \draw    (250.61,250.45) -- (203.39,207.18) ;
    \draw [shift={(201.91,205.83)}, rotate = 402.5] [color={rgb, 255:red, 0; green, 0; blue, 0 }  ][line width=0.75]    (6.56,-1.97) .. controls (4.17,-0.84) and (1.99,-0.18) .. (0,0) .. controls (1.99,0.18) and (4.17,0.84) .. (6.56,1.97)   ;
    \draw    (266.84,250.45) -- (219.62,207.18) ;
    \draw [shift={(218.15,205.83)}, rotate = 402.5] [color={rgb, 255:red, 0; green, 0; blue, 0 }  ][line width=0.75]    (6.56,-1.97) .. controls (4.17,-0.84) and (1.99,-0.18) .. (0,0) .. controls (1.99,0.18) and (4.17,0.84) .. (6.56,1.97)   ;
    \draw    (331.76,250.45) -- (331.76,207.83) ;
    \draw [shift={(331.76,205.83)}, rotate = 450] [color={rgb, 255:red, 0; green, 0; blue, 0 }  ][line width=0.75]    (6.56,-1.97) .. controls (4.17,-0.84) and (1.99,-0.18) .. (0,0) .. controls (1.99,0.18) and (4.17,0.84) .. (6.56,1.97)   ;
    \draw    (283.07,250.45) -- (283.07,207.83) ;
    \draw [shift={(283.07,205.83)}, rotate = 450] [color={rgb, 255:red, 0; green, 0; blue, 0 }  ][line width=0.75]    (6.56,-1.97) .. controls (4.17,-0.84) and (1.99,-0.18) .. (0,0) .. controls (1.99,0.18) and (4.17,0.84) .. (6.56,1.97)   ;
    \draw    (429.15,250.45) -- (429.15,207.83) ;
    \draw [shift={(429.15,205.83)}, rotate = 450] [color={rgb, 255:red, 0; green, 0; blue, 0 }  ][line width=0.75]    (6.56,-1.97) .. controls (4.17,-0.84) and (1.99,-0.18) .. (0,0) .. controls (1.99,0.18) and (4.17,0.84) .. (6.56,1.97)   ;
    \draw    (412.92,250.45) -- (397.37,207.71) ;
    \draw [shift={(396.69,205.83)}, rotate = 430.01] [color={rgb, 255:red, 0; green, 0; blue, 0 }  ][line width=0.75]    (6.56,-1.97) .. controls (4.17,-0.84) and (1.99,-0.18) .. (0,0) .. controls (1.99,0.18) and (4.17,0.84) .. (6.56,1.97)   ;       
  \end{tikzpicture}
    }
  \end{center}
\caption{The injective mapping from $L'$ to $H \cup L$. Each job in $L'$ is mapped to either a heavy job in $H$ (e.g., $a,b$) or it is mapped to itself (e.g., $x,y$, and $z$). }
\end{figure*}
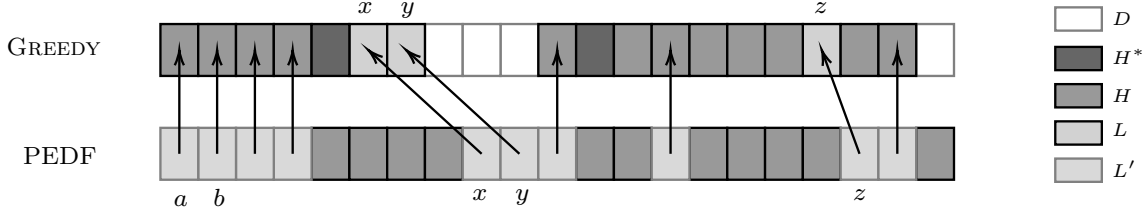

Note that the jobs executed by \textsc{PEDF} are $H \cup L'$ and the jobs executed by  \textsc{Greedy} are $H \cup H^* \cup L$.
We use the corresponding lower case notation for the cardinality of those sets.
The key to the proof is the following claim.

\begin{claim} \label{claim}
  $d+h^*\leq h$.
\end{claim}
The proof uses an injective mapping from $L'$ to $H\cup L$. This implies $\ell' \leq h+\ell$, and by the following equality, the claim holds:
\begin{equation} \label{eq:number_jobs}
  d+h^*+h+\ell = h + \ell'
\end{equation}
The left and right sides of the above equality respectively indicate the number of time slots throughout the execution of \textsc{Greedy} and \textsc{PEDF}.

To define the mapping, let $j\in L'$ be an arbitrary light job executed by \textsc{PEDF} at some moment $t$. If \textsc{Greedy} executes $j$ as well and not later than $t$, then we map $j$ to itself.  Otherwise, we map $j$ to the job $k$ executed by \textsc{Greedy} at time $t$.  We observe that if $k$ is a heavy job, then it is common to both schedules. This follows by the definition of \textsc{PEDF}, for which $k$ was pending at this point, hence $k\in H$.  To show the injective nature of the mapping, for the sake of contradiction, we assume that two jobs $j,k\in L'$ are mapped to the same job $k\in L$, executed by \textsc{Greedy} at some time $t$.  Then by the definition of the mapping, $k$ must be scheduled after $j$ by \textsc{PEDF}, so both jobs $j,k$ were pending at time $t$ for \textsc{PEDF} but also for \textsc{Greedy}.  This contradicts the fact that both schedules apply the same earliest deadline policy among the light jobs, and concludes the proof of the claim.

To summarize, we can express the objective values of the schedules as 
\begin{align*}
  \textrm{profit}(\textsc{OPT})    &\leq h^* \alpha + h \alpha  + \ell + d  \\
  \textrm{profit}(\textsc{Greedy}) &= h^* \alpha+  h \alpha + \ell\\
  \textrm{profit}(\textsc{PEDF})   &= h \alpha + \ell + h^* + d, \\
\end{align*}
where the last expression uses \eqref{eq:number_jobs}.

Claim~\ref{claim} implies \ifthenelse{\boolean{twocol}}{Inequality~(\ref{ineq:main})}{the following inequality}, which through a sequence of equivalent transformations transforms into the required bound on the competitive ratio\ifthenelse{\boolean{twocol}}{.}{:}

\ifthenelse{\boolean{twocol}}{\begin{figure*}[!h]\small}{\scriptsize}

 {
 \begin{align}
  \begin{split}
  h \alpha  + \ell & \geq  h^* \alpha  + d \alpha  
  \\
  h \alpha(\alpha-1)   + \ell (\alpha - 1)
&\geq
h^* \alpha (\alpha - 1) + d \alpha (\alpha - 1) 
  \\
  -(h \alpha + h^* \alpha + \ell )+
  \alpha (h \alpha + h^* + \ell + d)  
  &\geq
  h^* \alpha (\alpha - 1) + d \alpha^2 
  \\
  - (h \alpha + h^* \alpha + \ell)
  +   2 \alpha (h \alpha + h^* + \ell + d)  
  &\geq   
  \alpha(h \alpha + h^*+\ell+d) + h^* \alpha (\alpha-1)  + d \alpha^2 
  \\
  \alpha^2 (h \alpha + h^* \alpha + \ell)
  - (h \alpha + h^* \alpha + \ell)
  +
  2 \alpha (h \alpha + h^* + \ell + d)  
  &\geq   
  \alpha^2 (h \alpha + h^* \alpha + \ell) 
  + \alpha (h \alpha + h^* \alpha + \ell + d) +   d \alpha^2 
\\
  (\alpha^2 - 1) (h \alpha + h^* \alpha + \ell)
  + 
  2 \alpha (h \alpha + h^* + \ell + d)  
  &\geq   
  (\alpha^2 + \alpha) (h \alpha + h^* \alpha + \ell + d) 
  \\
  \frac{\alpha^2 - 1}{\alpha^2 + \alpha} (h \alpha + h^* \alpha + \ell) + \frac{2
    \alpha}{\alpha^2 + \alpha} (h \alpha + h^* + \ell + d) 
    & \geq  h \alpha + 
    h^* \alpha + \ell + d\\
    R \cdot p \cdot \textrm{profit}(\textsc{Greedy}) + R \cdot (1 - p)\cdot \textrm{profit}(\textsc{PEDF}) & \geq 
    \textrm{profit}(\textsc{Opt})\\  
    R \cdot \textrm{profit}(\textsc{Barely-Random}) & \geq  \textrm{profit}(\textsc{Opt}).
\end{split}\label{ineq:main}
\end{align}
}
\ifthenelse{\boolean{twocol}}{\end{figure*}}{}

\end{proof}

\section{Conclusion}

Now that the randomized competitive ratio of the two valued bounded delay
online buffer management problem is tackled --- against the oblivious adversary ---
it would be interesting to see how it can be generalized to more values for the job weights, for
example weights belonging to the set $\{1,\alpha,\alpha^2\}$, or even to
the general problem.  Currently the only known randomized algorithms for the
general problem are \textsc{Rmix} and \textsc{Remix}.

The algorithm \textsc{Rmix} was defined in \cite{bartal_online_2004,chin_online_2006} as follows. At
every time slot, consider $j$ the heaviest pending job. Denote $w_j$ its
weight. Let $x\in[-1,0]$ be a uniformly chosen real number. Let $f$ be a job
with minimal deadline and weight $w_f \geq e^x w_j$. Execute $f$ in this time
slot.  We can observe that, if executed on instance 1 of the lower bound instance from
Proposition~\ref{prop:rand_lower_bound}, at time $0$, it will choose the heavy
job with probability $\min\{1, \ln \alpha\}$.  Hence it has a larger tendency
to choose the heavy job than \textsc{Barely-Random}, and is fooled by instance
1.

The algorithm \textsc{ReMix} was defined in \cite{jez_universal_2013}, and
behaves even differently on instance 1. It will execute the heavy job with
probability $1-1/\alpha$ and the light job with probability $1/\alpha$, but is
also not optimal. These comparisons with \textsc{Barely-Random} are unfair in
a sense, because they have been designed for the general weighted case, and
mostly for an adaptive adversary. But the comparison indicates room for
improvement. The optimal randomized algorithm for the general weighted
problem, might combine ideas of all these mentioned algorithms, and in
particular borrow from the optimal deterministic algorithm
\cite{vesely_phi-competitive_2019}.

\section{Acknowledgment}

This work was partially supported by the French research agency (ANR-19-CE48-0016), the CNRS PEPS project ADVICE as well as by the EPSRC grant EP/S033483/1, and by the NSERC Discovery Grant. 

We would like to thank Spyros Angelopoulos for helpful discussions as well as anonymous referees for spotting an error in a previous version of this paper.


\begin{thebibliography}{10}
\expandafter\ifx\csname url\endcsname\relax
  \def\url#1{\texttt{#1}}\fi
\expandafter\ifx\csname urlprefix\endcsname\relax\def\urlprefix{URL }\fi
\expandafter\ifx\csname href\endcsname\relax
  \def\href#1#2{#2} \def\path#1{#1}\fi

\bibitem{hajek_competitiveness_2001}
B.~Hajek, On the competitiveness of on-line scheduling of unit-length packets
  with hard deadlines in slotted time, Proc. of the 2001 {Conference} on
  {Information} {Sciences} and {Systems}, 2001, pp. 434--439.

\bibitem{kesselman_buffer_2004}
A.~Kesselman, Z.~Lotker, Y.~Mansour, B.~Patt-Shamir, B.~Schieber,
  M.~Sviridenko, Buffer {Overflow} {Management} in {QoS} {Switches}, SIAM
  Journal on Computing 33~(3) (2004) 563--583.

\bibitem{vesely_phi-competitive_2019}
P.~Veselý, M.~Chrobak, Ł.~Jeż, J.~Sgall, A {$\Phi$}-{Competitive} {Algorithm}
  for {Scheduling} {Packets} with {Deadlines}, Proc. of the 2019 {Annual}
  {ACM}-{SIAM} {Symposium} on {Discrete} {Algorithms}, Society for Industrial
  and Applied Mathematics, 2019, pp. 123--142.

\bibitem{chin_online_2003}
F.~Y.~L. Chin, S.~P.~Y. Fung, Online {Scheduling} with {Partial} {Job}
  {Values}: {Does} {Timesharing} or {Randomization} {Help}?, Algorithmica
  37~(3) (2003) 149--164.

\bibitem{bartal_online_2004}
Y.~Bartal, F.~Y.~L. Chin, M.~Chrobak, S.~P.~Y. Fung, W.~Jawor, R.~Lavi,
  J.~Sgall, T.~Tichý, Online {Competitive} {Algorithms} for {Maximizing}
  {Weighted} {Throughput} of {Unit} {Jobs}, 
  Proc. of the 21th {Symposium} on {Theoretical} {Computer} {Science}, Lecture
  {Notes} in {Computer} {Science}, Springer, Berlin, Heidelberg, 2004, pp.
  187--198.

\bibitem{chin_online_2006}
F.~Y.~L. Chin, M.~Chrobak, S.~P.~Y. Fung, W.~Jawor, J.~Sgall, T.~Tichý, Online
  competitive algorithms for maximizing weighted throughput of unit jobs,
  Journal of Discrete Algorithms 4~(2) (2006) 255--276.

\bibitem{chin_improved_2004}
F.~Y.~L. Chin, S.~P.~Y. Fung, Improved competitive algorithms for online
  scheduling with partial job values, Theoretical Computer Science 325~(3)
  (2004) 467--478.

\bibitem{goldwasser_survey_2010}
M.~H. Goldwasser, A survey of buffer management policies for packet switches,
  ACM SIGACT News 41~(1) (2010) 100--128.

\bibitem{englert_considering_2012}
M.~Englert, M.~Westermann, Considering {Suppressed} {Packets} {Improves}
  {Buffer} {Management} in {Quality} of {Service} {Switches}, SIAM Journal on
  Computing 41~(5) (2012) 1166--1192.

\bibitem{borodin_online_2005}
A.~Borodin, R.~El-Yaniv, Online computation and competitive analysis, Cambridge
  University Press, 2005.

\bibitem{jez_universal_2013}
Ł.~Jeż, A {Universal} {Randomized} {Packet} {Scheduling} {Algorithm},
  Algorithmica 67~(4) (2013) 498--515.

\end{thebibliography}
\end{document}